\documentclass[lettersize,journal]{IEEEtran}
\usepackage{cite}
\usepackage{amsmath,amssymb,amsfonts}
\usepackage{graphicx}
\usepackage{textcomp}
\usepackage{amsthm}

\usepackage{color}
\usepackage{graphics,epsfig,subfigure}
\usepackage{bm,mathtools}
\usepackage{multirow,multicol}
\usepackage{algorithm,algpseudocode}
\newtheorem{theorem}{Theorem~}
\newtheorem{assumption}{Assumption~}
\newtheorem{remark}{Remark~}
\newtheorem{lemma}{Lemma~}
\newtheorem{problem}{Problem~}
\begin{document}
	\title{Coalitional Zero-Sum Games for ${H_{\infty}}$ Leader-Following Consensus Control}
	
	\author{Yunxiao~Ren, Dingguo Liang$^\star$,  Yuezu~Lv,~\IEEEmembership{Senior Member,~IEEE,}  
	 Zhisheng~Duan,~\IEEEmembership{Senior Member,~IEEE} 
		\thanks{Y.~Ren is with the College of Engineering, Peking University, Beijing 100871, China (email: renyx@pku.edu.cn).}
		\thanks{D. Liang is with Institute for Automatic Control and Complex Systems, University of Duisburg-Essen, Duisburg 47057, Germany (corresponding author; email: dingguo.liang@uni-due.de).}
		\thanks{Y.~Lv is with  the Beijing Key Laboratory of
			Lightweight Intelligent System, Beijing Institute of Technology, Beijing, 100081, China(email:yzlv@bit.edu.cn).}
		\thanks{Z.~Duan is with the School of advanced manufacturing and robotics , Peking University, Beijing 100871, China (email:duanzs@pku.edu.cn).}
		\thanks{This work has been submitted to the IEEE for possible publication. Copyright may be transferred without notice, after which this version may no longer be accessible.}
	}

	\maketitle
	

		\begin{abstract}
	This paper investigates the leader-following consensus problem for a class of multi-agent systems subject to adversarial attack-like external inputs. To address this, we formulate the robust leader-following control problem as a global coalitional min-max zero-sum game using differential game theory. Specifically, the agents' control inputs form a coalition to minimize a global cost function, while the attacks form an opposing coalition to maximize it. Notably, when these external adversarial attacks manifest as disturbances, the designed game-theoretic control policy systematically yields a robust $H_\infty$ control law. Addressing this problem inherently requires solving a high-dimensional generalized algebraic Riccati equation (GARE), which poses significant challenges for distributed computation and controller implementation. To overcome these challenges, we propose a two-fold approach. First, a decentralized computational strategy is devised to decompose the high-dimensional GARE into multiple uniform, lower-dimensional GAREs.  Second, a dynamic average consensus-based decoupling algorithm is developed to resolve the inherent coupling structure of the robust control law, thereby facilitating its distributed implementation. Finally, numerical simulations on the formation control of multi-vehicle systems with feedback-linearized dynamics are conducted to validate the effectiveness of the proposed algorithms.
		\end{abstract}
	\begin{IEEEkeywords}
		multi-agent systems, robust control,  leader-follower consensus, min-max game. 
	\end{IEEEkeywords}
	\IEEEpeerreviewmaketitle

	
\section{Introduction}
In recent decades, the consensus control of multi-agent systems (MASs) has garnered substantial attention, primarily due to its versatile applicability across diverse engineering domains. Notable applications include multi-vehicle cooperation \cite{4026052}, distributed fault detection \cite{Liang2021JFI}, and sensor networks \cite{10947572}. The core objective of this control paradigm is to achieve state synchronization among agents by establishing agreement on critical variables, relying exclusively on local neighbor information. The consensus framework can be fundamentally categorized into leaderless consensus \cite{li2009consensus} and leader-following consensus \cite{Wen2}, distinguished by the presence of a virtual or real leader. Notably, leader-following consensus has exhibited broader applicability, as it enables the intentional design of the leader’s trajectory to guide the collaborative motion of the entire system. 

However, with the rapid proliferation of cyber-physical systems (CPSs), MASs are increasingly vulnerable to adversarial, ``attack-like" external inputs \cite{11090164, 11143595}. Unlike stochastic disturbances, these adversarial inputs are intentionally synthesized to intelligently degrade system performance, posing profound security and resilience challenges. To effectively capture the strategic nature of this confrontation, game-theoretic approaches can be employed to formulate the interaction between the system's controllers and the worst-case attacks as a min-max zero-sum game \cite{bacsar2008h}. Notably, under appropriate formulations, the saddle-point solution of this zero-sum game systematically yields a robust $H_{\infty}$ control policy. Consequently, this game-theoretic paradigm emerges as an effective approach to resiliently handle these severe external perturbations, ensuring both prescribed disturbance attenuation and closed-loop stability \cite{zhou1996robust}.

\subsection{Related Works}Numerous studies have explored the $H_{\infty}$ consensus control of multi-agent systems. For instance, in \cite{Li2009Hinf}, a decentralized $H_{\infty}$ controller for networked MASs is designed using linear matrix inequalities (LMIs). In \cite{Liu2011,doi:10.1080/00207170903267039}, LMI-based conditions are derived to ensure MAS consensus with a prescribed $H_{\infty}$ attenuation level. The finite $L_{2}$-gain performance index for nonlinear MASs is addressed in \cite{Wen2012nonlinearL2}, while the $H_{\infty}$ consensus problem for linear agents under directed communication graphs is solved in \cite{WangJY2013HinfandH2}. Traditionally, these works have aimed to bound the $H_{\infty}$ norm of transfer functions from disturbances to performance outputs using conventional robust control techniques. More recently, there has been a surging interest in applying differential game theory to address optimal and robust consensus problems \cite{Vamvoudakis2012, Jiao2016, Adib2019Hinf,RENNSE,Zhouyan_minmax, zhu2025game}. In this game-theoretic context, the considered attacks act as adversarial players possessing greater intelligence and the capacity to maximize the cost function, thereby undermining control performance. These disruptive elements, often termed ``attack-like disturbances," are characterized by their strategic actions and are highly prevalent in the modern CPS landscape. Differential games have proven highly effective in such scenarios. For instance, differential graphical games are utilized in \cite{Vamvoudakis2012} for online adaptive learning synchronization. Zero-sum games and min-max optimization problems are widely employed to formulate robust $H_{\infty}$ control problems \cite{bacsar2008h, REN2022IJRNC}. Moreover, disturbance rejection via multi-agent zero-sum differential graphical games is proposed in \cite{Jiao2016}, and state synchronization for linear homogeneous agents is addressed in \cite{Adib2019Hinf}. In \cite{Zhouyan_minmax}, a distributed min-max strategy is designed for consensus tracking-control in the presence of external attacks. Despite these advancements, existing graphical game-based methodologies \cite{Jiao2016, Adib2019Hinf,Zhouyan_minmax} inherently achieve only local Nash equilibriums rather than a global optimal point, owing to their reliance on localized cost functions. Such formulations prevent the control inputs of individual agents from fully cooperating to minimize the global cost, ultimately resulting in conservative control policies. Given these limitations, it is highly important to investigate the broader global coalitional min-max game problem. In this paradigm, the control inputs of all agents collaborate as a unified coalition to minimize a global cost function, while the attacks form a counter-coalition striving to maximize it. Although the global linear quadratic regulator (LQR) problem for MASs has been studied \cite{zhangzhuo2021TAC}, the sheer complexity of the coupled weighting matrices makes performance tuning exceedingly difficult for practical applications. Consequently, a significant gap remains in efficiently solving the global coalitional min-max game for multi-agent systems, primarily due to the severe computational and implementational bottlenecks caused by structural coupling.
\subsection{Contributions}Motivated by the aforementioned challenges, this paper proposes an innovative framework featuring a decentralized computational strategy and a distributed information-fusion decoupling algorithm to elegantly solve the global coalitional min-max optimization problem. The core contributions of this paper are summarized as follows:
\begin{enumerate}
	\item[1)] This paper formulates the robust leader-following control of MASs under adversarial attacks as a global coalitional min-max game. By driving the agents' control inputs to collaboratively minimize the global cost function against the adversarial inputs, our methodology achieves true global optimization. This fundamentally differs from the previous graphical game approaches \cite{Jiao2016, Adib2019Hinf} that are constrained to local Nash equilibrium solutions. Notably, when these adversarial attacks manifest as disturbances, optimizing the agents' collaborative control inputs within this game-theoretic framework systematically yields a global robust $H_\infty$ control policy.
	\item[2)] The designed cost function features a straightforward and user-friendly weighting matrix structure. This is in stark contrast to the heavily coupled and intricate weighting matrices required in existing literature \cite{zhangzhuo2021TAC}. This structural simplicity empowers users to intuitively tune the cost weights, seamlessly facilitating desired system performance objectives such as accelerated tracking or enhanced energy efficiency.
	\item[3)] A decentralized computation strategy combined with a distributed decoupling algorithm is proposed to completely separate the computation and implementation of the control laws. The inherently complex, high-dimensional coupled GARE is systematically decomposed into multiple independent, lower-dimensional GAREs. The proposed decoupling algorithm is rigorously proven to converge in finite time, ensuring rapid deployment. By reducing the global high-dimensional GARE into uniform low-dimensional ones, the algorithm ensures that each agent’s computational load remains constant regardless of the total population $N$, ensuring topological universality and computational efficiency in dense networks. \end{enumerate}	
\section{Preliminaries and problem formulation}
\subsection{Basic Graph Theory}
In this section, we will review some basic concepts from graph theory. A weighted graph with $N$ nodes can be denoted as $\mathcal{G} = \left\{\mathcal{N},\mathcal{E}\right\}$, where $\mathcal{N}$ is the vertex set with $i \in \mathcal{N}$ corresponding to node $i$, and $\mathcal{E}$ is the edge set of the graph, which is a subset of $\mathcal{N}\times \mathcal{N}$. The adjacency matrix $E$ of the graph is defined as $\{a_{ij}\}$, where $a_{ii} = 0$ and $a_{ij}\neq0$ if and only if $(j,i) \in \mathcal{E}$. The in-neighborhood set of a node $i$ is denoted by $\mathcal{N}_{i} = \left\{j \in \mathcal{N} : (j, i) \in \mathcal{E}, i \not= j\right\}$. The graph Laplacian matrix is defined as $L = D- E$, where $D$ is the in-degree matrix with $d_{i} = \sum_{j \in \mathcal{N}_{i}} a_{ij}$ being the in-degree of node $i$. A path on $\mathcal{G}$ from node $i_{1}$ to node $i_{l}$ is a sequence of ordered pairs $(i_{k}, i_{k+1})$, where $k = 1, \cdots, l-1$.
	\subsection{Problem Formulation}
	Consider $N$ agents distributed over communication by graph $\mathcal{G}$ with dynamics as the following form $i  \in \mathcal{N} =  \{1,\cdots,N\}$
	\begin{equation}
		\dot{x}_{i}(t) = Ax_{i}(t)+Bu_{i}(t)+Dw_{i}(t),
		\label{eq:dyn}
	\end{equation}
	where $x_{i}(t) \in \mathbb{R}^{n}$ is the state of node $i$, $u_{i}\left(t\right) \in \mathbb{R}^{m_{1}}$ is the control input of node $i$, and $w_{i} \in \mathbb{R}^{m_{2}}$ is the adversarial attacks of node $i$, $i \in \mathcal{N}$. And $A\in \mathbb{R}^{n \times n}$, $B \in \mathbb{R}^{n\times m_{1}}$, $D\in \mathbb{R}^{n\times m_{2}}$  are the system, input and attack input matrices, respectively.\par
	The leader node's dynamics is 
	\begin{equation}
		\label{eq:ldyn}
		\dot{x}_{0}(t) = Ax_{0}(t).
	\end{equation}\par
	The neighbor error for each agent is defined as
	\begin{equation}
		\label{eq:error}
		\delta_{i}:=\sum_{j \in \mathcal{N}_{i}} a_{i j}\left(x_{i}-x_{j}\right)+g_{i}\left(x_{i}-x_{0}\right), \quad \forall i \in \mathcal{N},
	\end{equation}
	where $g_{i}$ is the pinning gain of agent $i$, $g_{i} \neq 0$ indicates that agent $i$ is pinned to the leader node, $\delta_{i} \in \mathbb{R}^{n}$. From (\ref{eq:error}), the over all neighbor error vector is given by
	\begin{equation}\label{eq:overallerror}
		\delta=\left((L+G) \otimes I_{n}\right)\xi
	\end{equation}
	where  $x=\left[
	x_{1}^{T} , x_{2}^{T} , \ldots x_{N}^{T}
	\right]^{T} $, $\delta=\left[
	\delta_{1}^{T} , \delta_{2}^{T} , \ldots , \delta_{N}^{T}
	\right]^{T} \in \mathbb{R}^{n N}$, $\underline{x}_{0} = \mathbf{1}_{N}\otimes x_{0} \in \mathbb{R}^{nN}$, $\xi =\left(x-\underline{x}_{0}\right)$ is the synchronization error, $L$ is the Laplacian of the graph and $G = \text{diag}\{g_{1},g_{2},\ldots,g_{N}\}$,
	and the over all form of neighbor error dynamics is
	\begin{equation}\label{eq:overallerrdyn}
		\dot{\delta} = (I_{N}\otimes A)\delta+ [(L+G)\otimes B]u+[(L+G)\otimes D]w,
	\end{equation}
	where $u = \left[u_{1}^{T}, u_{2}^{T},\ldots ,u_{N}^{T}\right]^{T}$ and $w = \left[w_{1}^{T}, w_{2}^{T},\ldots ,w_{N}^{T}\right]^{T}$.
	\begin{assumption}
		\label{assum:1}
		\begin{enumerate}\item[]
			\item[1)] $(A,B)$ is controllable .
			\item [2)] The graph $\mathcal{G}$ is undirected, connected and at least one pinning gain $g_{i}$  is nonzero.
		\end{enumerate}
	\end{assumption}\par
	
	\begin{remark}
		From Assumption \ref{assum:1}, one can get that $(L+G)$ is symmetric and  nonsingular \cite{Khoo2009TOM}. That is to say, $\|\xi\|_{2} \leq \bar{\sigma}((L+G)^{-1})\|\delta\|_{2}$  such that $\lim\limits_{t \to \infty }\| x_{i}\left( t\right)-x_{0}\left(t\right)  \|_{2}= 0$ if and only if $ \lim\limits_{t \to \infty }\| \delta_{i}\|_{2} =0,$ for  $\forall i \in \mathcal{N}$. 
	\end{remark}
	\begin{assumption}\label{assum:N}
		Each agent $i$ knows the network size $N$, and an upper bound of the network size is $\bar{N}$.
	\end{assumption}
	\begin{assumption}\label{assum:Li}
		Each agent $i$ can access the corresponding row of the Laplacian matrix $L+G$, i.e., agent $i$ knows $(L+G)_{i}$.
	\end{assumption}
		\begin{assumption} \label{assum:bound}
		The time derivative of the tracking error for each agent is bounded, i.e., $\|\dot{\delta_{i}}\|_{2} \leq \eta_{i}$, where $\eta_{i}$ is a positive real number.
	\end{assumption}

	\begin{remark}
Assumption \ref{assum:N} is well-founded, as there exist numerous distributed techniques capable of estimating the network size swiftly (see \cite{shim2019CDC} and references therein). Assumption \ref{assum:Li} is also justifiable. In the network, every agent is assigned a unique label, and the transmitted information between agents retains these labels. Consequently, each agent has access to the labels of all its neighboring agents. Thus, since each agent is aware of the network size $N$, it can deduce $(L+G)_{i}$ for its own corresponding row. 	Assumption \ref{assum:bound} imposes a standard requirement that the time derivative of the tracking error for each agent remains bounded. Similar assumptions can be found in literature, such as \cite{chenfei2012TAC}.
	\end{remark}

The objective of this paper is to address the robust leader-following consensus problem in the global coalitional min-max game framework by designing an appropriate controller in a distributed manner.

Consider the following quadratic performance function:

\begin{equation}
	J(\delta,u,w) = \frac{1}{2}\int_{0}^{\infty}(\delta^{T}Q\delta + u^{T}Ru - w^{T}\Gamma w)dt,
\end{equation}

where $Q \geq 0$, $(Q,I_{N}\otimes A)$ is observable, and $R > 0$ and $\Gamma > 0$ are weighting matrices to be defined.

For state feedback strategies, the corresponding value function is:

\begin{equation}
	V(\delta(t)) = \frac{1}{2}\int_{t}^{\infty}(\delta^{T}Q\delta + u^{T}Ru - w^{T}\Gamma w)dt.
\end{equation}

\begin{problem}(Distributed Global Coalitional Min-Max Game):
	The objective is to find the distributed optimal control strategies $u_{i}^{*}$ and worst-case attack strategies $w_{i}^{*}$ for each agent $i \in \{1, \dots, N\}$ that solve the following global coalitional min-max game problem:
	\begin{equation}\label{eq:minmax-prob}
		V^{*}(\delta(t)) = \min_{u}\max_{w} V(\delta(t)),
	\end{equation}
subject to the overall error dynamics in \eqref{eq:overallerrdyn}, under the strict constraint of distributed information structures.

Specifically, to solve this global optimization problem distributedly, the proposed control protocol $u_{i}$ for each agent $i$ must satisfy the following conditions:
\begin{enumerate}
\item[(1)]Local Implementation: The control law $u_{i}$ can only utilize local state information $\delta_{i}$ and information from its immediate neighbors defined by the communication graph $\mathcal{G}$.
\item[(2)]Decentralized Computation: The derivation of the control parameters (e.g., solving the associated high-dimensional generalized algebraic Riccati equation) must be decoupled into low-dimensional, local computations without relying on a centralized coordinator.
\end{enumerate}
\end{problem}

To provide a more intuitive understanding of the proposed theoretical framework, the architectural paradigm of the global coalitional zero-sum game is illustrated in Fig. \ref{fig:concept_game}. As depicted, the multi-agent system, guided by the leader node via the communication graph $\mathcal{G}$, acts as the physical interactive environment. To systematically address the system's vulnerability, all local control inputs $u_i$ are conceptually aggregated into a unified control coalition (highlighted in blue) aiming to minimize the global cost $J$. Conversely, the external attack inputs $w_i$ form an opposing adversarial coalition (highlighted in red) striving to maximize the same cost. The saddle-point solution of this global min-max game ultimately dictates the robust consensus policy.

\begin{figure*}[htbp]
	\centering
	\includegraphics[width=0.95\linewidth]{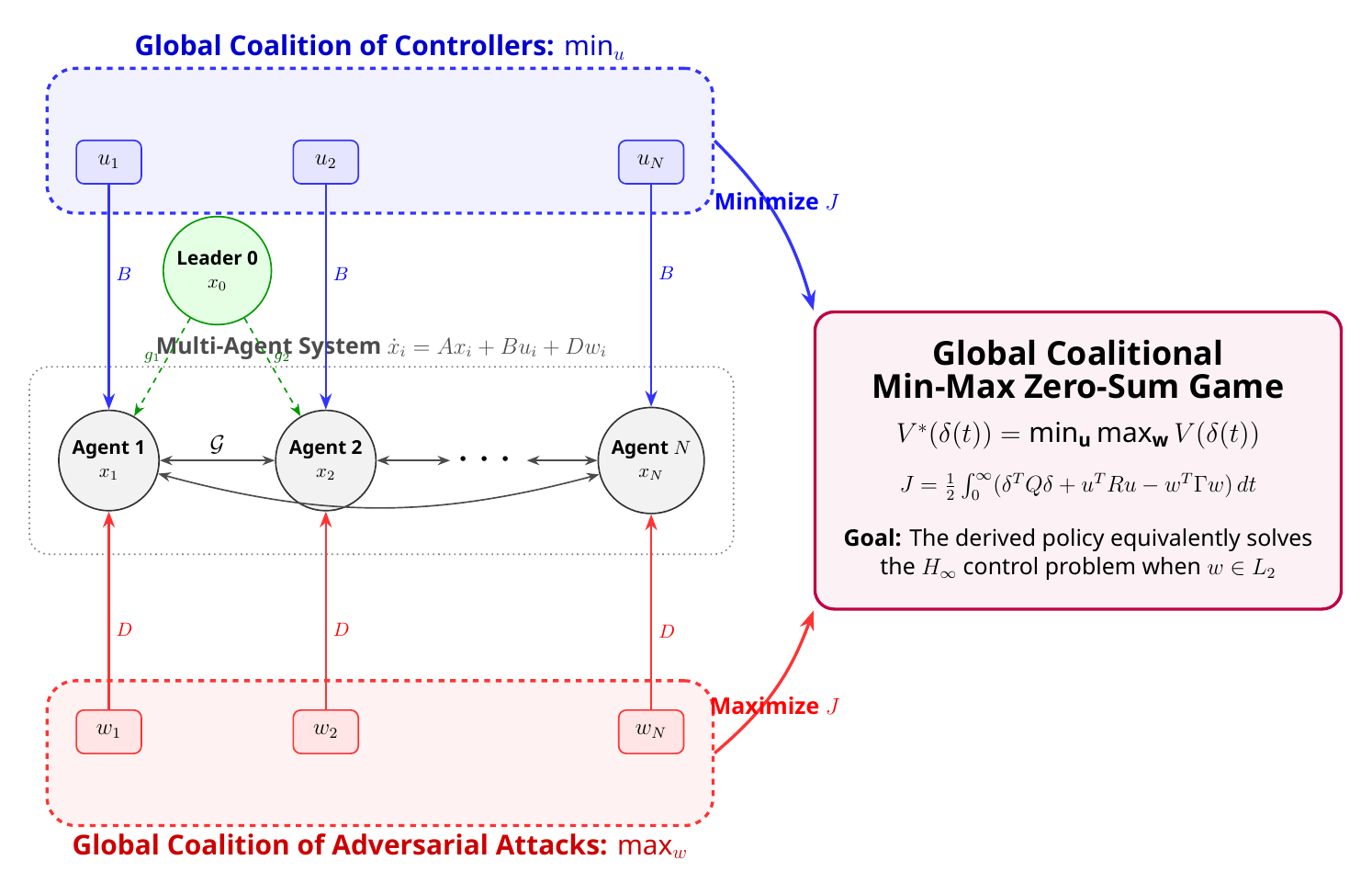} 
	\caption{Conceptual diagram of the global coalitional min-max zero-sum game framework for multi-agent systems, illustrating the structural game relationship between the controllers and the adversarial attacks.}
	\label{fig:concept_game}
\end{figure*}

By treating the worst-case adversarial attacks as external $L_2$ disturbances, solving the aforementioned min-max zero-sum game is mathematically equivalent to addressing the following bounded $L_2$-gain synchronization problem:

\begin{problem}(Bounded $L_{2}$-Gain Synchronization Problem)
	Consider system (\ref{eq:overallerrdyn}) with a concerned performance output $z$ and external $L_2$ disturbance $w$. The bounded $L_{2}$-gain synchronization problem aims to design the control input $u$ such that:
	\begin{enumerate}
		\item[(1)] {Asymptotic Synchronization (Internal Stability):} In the absence of external disturbances (i.e., $w \equiv 0$), the neighborhood tracking error $\delta(t)$ asymptotically converges to zero as $t \to \infty$.
		\item[(2)] {Robust Disturbance Attenuation ($H_{\infty}$ Performance):} For any non-zero, square-integrable external disturbance $w(t) \in L_{2}[0, T]$, the system satisfies the following bounded $L_{2}$-gain condition:
		\begin{equation}\label{eq:bounded_L2}
			\int_{0}^{T} \|z(t)\|_{M}^{2} dt \leq \gamma^{2} \int_{0}^{T} \|w(t)\|_{H}^{2} dt + \alpha(\delta(0)),
		\end{equation}
		where $M$ and $H$ are user-defined positive definite weighting matrices, $\gamma > 0$ denotes the prescribed disturbance attenuation level, and $\alpha(\cdot)$ is a non-negative continuous scalar function satisfying $\alpha(0) = 0$.
	\end{enumerate}
\end{problem}

\section{The Global Coalitional Min-max Strategies}
In this section, we propose the solution to the Global Coalitional Min-Max Game. Because the system dynamics described in \eqref{eq:overallerrdyn} are linear, the optimal solution is provided in the following lemma:
	\begin{lemma}
		The following  are the min-max strategies corresponding to \eqref{eq:minmax-prob}
		\begin{equation}\label{eq:u*}
			u^{*} = -R^{-1}\bar{B}^{T}P \delta,
		\end{equation}
		\begin{equation}\label{eq:w*}
			w^{*} = \Gamma^{-1}\bar{D}^{T}P \delta,
		\end{equation}
		where $\bar{A} = I_{N}\otimes A$, $\bar{B} = (L+G)\otimes B$, $\bar{D} = (L+G)\otimes D$, and $P$ is the unique positive-definite soltion of the following generalized algebraic Riccati equation 
		\begin{equation}\label{eq:Riccatioverall}
			\bar{A}^{T}P+P\bar{A}+Q- P\bar{B}R^{-1}\bar{B}^{T}P+P\bar{D}\Gamma^{-1}\bar{D}^{T}P = 0.
		\end{equation}
	\end{lemma} 
	\begin{proof}
		According to the Bellman optimality principle, the following Bellman equation is obtained
		\begin{equation}\label{eq:HJB}
			\begin{aligned}
				H\left(\delta,u,w\right) &= \left(\frac{\partial V}{\partial \delta}\right)^{T}\left(\bar{A}\delta+\bar{B}u+\bar{D}w\right)\\&+\frac{1}{2}\left(\delta^{T}Q\delta+u^{T}Ru-w^{T}\Gamma w\right) =0.
			\end{aligned}
		\end{equation}
		For state feedback strategies, the value function has the following quadratic form
		\begin{equation}\label{eq:Vquad}
			V\left(\delta(t)\right) =\frac{1}{2} \delta^{T}(t)P\delta(t).
		\end{equation}
		The optimal control strategy $u^{*}$ and the worst-case attack $w^{*}$ satisfy the necessary condition $\frac{\partial H}{\partial u}=0$ and $\frac{\partial H}{\partial w}=0$, respectively. Then, one has $u^{*} = -R^{-1}\bar{B}^{T}P \delta$, $w^{*} = \Gamma^{-1}\bar{D}^{T}P \delta$, which are the same with \eqref{eq:u*} and \eqref{eq:w*}.
		
		Substitute \eqref{eq:u*} and \eqref{eq:w*} into \eqref{eq:HJB}, the GARE \eqref{eq:Riccatioverall} can be obtained.
	\end{proof}
	\begin{remark}
		The generalized algebraic Riccati equation is often occured in solving the standard $H_{\infty}$ problem. To guarantee it has a unique positive-definite solution, $(Q,\bar{A})$ needs to be observable, $(\bar{A},-\bar{B}R^{-1}\bar{B}^{T}+\bar{D}\Gamma^{-1}\bar{D}^{T})$ needs to be stabilizable, and $\bar{B}R^{-1}\bar{B}^{T}-\bar{D}\Gamma^{-1}\bar{D}^{T}$ should be positive semi-definite. These condition can be satisfied by choosing large $\Gamma$.
	\end{remark}
	
	With above lemma, it will be shown that the controller solves the bounded $L_{2}$-gain synchronization problem. 
	
	\begin{theorem}
		The controller \eqref{eq:u*} solves the bounded $L_{2}$-gain problem.
	\end{theorem}
	
	\begin{proof}
		First, it shows that when $w=0$, the systems achieve synchronization. Substitute the controller \eqref{eq:u*} and $w=0$ into \eqref{eq:overallerrdyn}, one has that the closed-loop dynamics
		\begin{equation}
			\dot{\delta} = \left(\bar{A}-\bar{B}R^{-1}\bar{B}^{T}P\right)\delta = \bar{A}_c \delta.
		\end{equation} 
		From \eqref{eq:Riccatioverall}, it can be obtained
		\begin{equation}
			\bar{A}_{c}^{T}P+P\bar{A}_{c} + P\bar{B}R^{-1}\bar{B}^{T}P+P\bar{D}\Gamma^{-1}\bar{D}^{T}P+Q = 0,
		\end{equation}  
		which is a Lyapunov equaiton, and means $\bar{A}_{c}$ is Hurwitz. That is to say, the neighbour error converges to zero asymptotically. 
		
		Then, it is proven that the bounded $L_{2}$-gain condition \eqref{eq:bounded_L2} is satisfied when $w=0$. Substitute \eqref{eq:u*} and \eqref{eq:w*} into \eqref{eq:HJB}, and consider \eqref{eq:Vquad},  one has 
		\begin{equation}\label{eq:HJB*}
			\begin{aligned}
			\left(\frac{\partial V}{\partial \delta}\right)^{T}&\left(\bar{A}\delta+\bar{B}u^{*}+\bar{D}w^{*}\right)\\&+\frac{1}{2}\left(\delta^{T}Q\delta+(u^*)^{T}Ru^{*}-(w^{*})^{T}\Gamma w^{*}\right) =0.
			\end{aligned}
		\end{equation}
	\end{proof}
	Consider the error dynamics driven by $u^{*}$ and $w\neq0$
	\begin{equation}
		\dot{\delta} = \bar{A}\delta+\bar{B}u^{*}+\bar{D}w,
	\end{equation}
	one can obtained that 
	\begin{equation}
		\frac{d V}{d t} = \left(\frac{\partial V}{\partial \delta}\right)^{T}\left(\bar{A}\delta+\bar{B}u^{*}+\bar{D}w\right).
	\end{equation}
	
	Therefore from \eqref{eq:HJB*}, it follows that 
	\begin{equation}\label{eq:L2ineq}
		\begin{aligned}
			\frac{d V}{d t}  =& -\frac{1}{2}\left(\delta^{T}Q\delta+(u^*)^{T}Ru^{*}-w^{T}\Gamma w\right) \\&-\frac{1}{2}\left(w^{*}-w\right)^{T}\Gamma\left(w^{*}-w\right)
			\\\leq &-\frac{1}{2}\left(\delta^{T}Q\delta+(u^*)^{T}Ru^{*}-w^{T}\Gamma w\right).
		\end{aligned}
	\end{equation}
	
	Integrate \eqref{eq:L2ineq} from $0\to T$, it follows that
	\begin{equation}\begin{aligned}
					V&(\delta(T)) - V(\delta(0)) \\&+\frac{1}{2}\int_{0}^{T}\left(\delta^{T}Q\delta+(u^*)^{T}Ru^{*}-w^{T}\Gamma w\right)\leq 0.
		\end{aligned}
	\end{equation}
	Since $V(\delta(T))>0$, one has 
	\begin{equation}
		\int_{0}^{T}\left(\delta^{T}Q\delta+(u^*)^{T}Ru^{*}\right)\leq\int_{0}^{T}\left( w^{T}\Gamma w\right)+2V(\delta(0)).
	\end{equation}
	Let $z = \left(\delta^{T}, u^{T}\right)^{T}$, $M = diag\{Q,R\}$, $\gamma = \sqrt{\bar{\sigma}(\Gamma)}$, and $\alpha(\cdot) = 2V(\cdot)$, it can be obtained that the bounded $L_{2}$-gain condition \eqref{eq:bounded_L2} is satisfied.
	\begin{remark}
		In summary, Theorem 1 derives a robust optimal control policy that ensures prescribed $H_\infty$ performance against disturbances. Despite this significant theoretical advantage, the practical application of these centralized strategies faces two severe bottlenecks. First, computing the control gains relies heavily on a centralized information structure. Specifically, solving the high-dimensional GARE \eqref{eq:Riccatioverall} necessitates complete knowledge of the global communication topology and the weighting matrices of all agents. More critically, the global GARE involves system matrices of dimension $Nn \times Nn$. Since the computational complexity of solving a Riccati equation typically scales cubically with the matrix dimension (i.e., $\mathcal{O}(N^3n^3)$), the computational burden suffers from a severe dimensionality explosion as the number of agents $N$ increases. This curse of dimensionality renders the centralized computation completely intractable for large-scale networks. Second, the resulting optimal control laws exhibit inherent structural coupling among neighboring agents, preventing independent execution. These constraints severely limit the scalability and autonomy of the multi-agent system. Consequently, overcoming these computational and implementational barriers will be the primary focus of the subsequent section.
	\end{remark}
	
	\section{Decomposition of the Strategies}
To overcome the centralization and coupling challenges identified in the previous section, this section proposes a comprehensive framework for the decentralized computation and distributed implementation of the global min-max strategies.

The proposed methodology is two-fold. First, a decentralized computational approach is developed by systematically designing the structures of the global weighting matrices. This strategic formulation decomposes the complex, high-dimensional GARE into multiple uniform, low-dimensional GAREs, empowering each agent to compute its control parameters independently without relying on global network information. Furthermore, to address the execution bottleneck, a dynamic average consensus-based algorithm is introduced to resolve the inherent structural coupling of the optimal control laws. By facilitating real-time state decoupling, this algorithm enables fully distributed execution, thereby significantly promoting the scalability and robustness of the overall system.

The remainder of this section details this decoupling methodology, beginning with the computation decomposition strategy in the following subsection.
	\subsection{Compuatation decomposition}
Note that the GARE \eqref{eq:Riccatioverall} contains the complete graph topology information, which is global information that is challenging for each agent to obtain. Therefore, the GARE needs to be decomposed in order to make the computation implementable. Therefore, the following weighting matrices are proposed
	\begin{equation}\label{eq:weightQ}
		Q = diag\{Q_{1}, Q_{2}, \ldots,Q_{N}\},
	\end{equation}
	\begin{equation}\label{eq:weightR}
		R  =[(L+G)\otimes I_{m_{1}}]^{T}\bar{R}[(L+G)\otimes I_{m_{1}}],
	\end{equation}
	\begin{equation}\label{eq:weightGamma}
		\Gamma =[(L+G)\otimes I_{m_{2}}]^{T}\bar{\Gamma}[(L+G)\otimes I_{m_{2}}],
	\end{equation}
	where $\bar{R} = diag\{R_{1}, \ldots,R_{N}\}$, $\bar{\Gamma} = diag\{\gamma^{2}_{1}I_{m_{2}},  \ldots,\gamma^{2}_{N}I_{m_{2}}\}$ $Q_{i}\geq 0$, $R_{i}>0$ are the local weighting matrices chosen by each agent, $\gamma_{i} >0$ is a parameter associated with the $L_{2}$-gain, $i =1, 2, \ldots N$. 
	
It is important to highlight that although the weighting matrices contain the graph information in the form of $L+G$, this information is not directly used in the computation process. The purpose of incorporating $L+G$ in the weighting matrices is to capture the overall network structure and facilitate decomposition, rather than to perform calculations involving $L+G$ itself.

Furthermore, it should be noted that $(L+G)$ is positive definite, which implies that tuning the local weighting matrices will correspondingly affect the global weighting. Increasing or decreasing the local weighting matrices will scale up or down the global weighting, respectively. This characteristic allows for flexibility in adjusting the impact of local behaviors on the overall system performance.

Then, the following result provides the decomposition of the strategies that solve the min-max problem \eqref{eq:minmax-prob}.
	
	\begin{theorem}
		The equations \eqref{eq:ucoupled} and \eqref{eq:wcoupled} are the optimal local strategy and the worst-case attack of agent $i$ to solve the global coalitional min-max problem \eqref{eq:minmax-prob} with weighting matrces choosing as \eqref{eq:weightQ}, \eqref{eq:weightR} and \eqref{eq:weightGamma}.
		\begin{equation}\label{eq:ucoupled}
			u^{*}_{i}=\frac{\left( \sum_{j=1}^{N} a_{i j} u^{*}_{j}-K_{i} \delta_{i}\right)}{\left( \sum_{j=1}^{N} a_{i j}+g_{i}\right)}, \forall i \in \mathcal{N},
		\end{equation}
		\begin{equation}\label{eq:wcoupled}
			w^{*}_{i}=\frac{\left( \sum_{j=1}^{N} a_{i j} u^{*}_{j}+L_{i} \delta_{i}\right)}{\left( \sum_{j=1}^{N} a_{i j}+g_{i}\right)}, \forall i \in \mathcal{N},
		\end{equation}
		where $K_{i} = R_{i}^{-1}B^{T}P_{i}$, $L_{i} = 1/\gamma_{i}^{2}D^{T}P_{i}$ and $P_{i}=P_{i}^{T}>0$ is the solution of the following local GARE:
		\begin{equation}\label{eq:localGARE}
			P_{i} A+A^{T} P_{i}+Q_{i}-P_{i} B R_{i}^{-1} B^{T} P_{i}+1/\gamma_{i}^{2}P_{i}DD^{T}P_{i}=0.
		\end{equation}
	\end{theorem}
	The obsevability of $(Q_{i}^{\frac{1}{2}}, A)$  and the stabilizability of $(A,-BR_{i}^{-1}B^{T}+1/\gamma_{i}^{2}DD^{T})$  guarantee that the local GARE (\ref{eq:localGARE}) has a positive definite solution $P_{i}$.
	
	\begin{proof}
		Substituting the weighting matrices \eqref{eq:weightQ}, \eqref{eq:weightR} and \eqref{eq:weightGamma} into the global GARE \eqref{eq:Riccatioverall}, and considering 
		\begin{equation}
		\begin{aligned}
			\bar{B}R^{-1}\bar{B} &= ((L+G)\otimes B)((L+G)\otimes I_{m})^{-1}\bar{R}^{-1}\\&\times((L+G)\otimes I_{m})^{-T}((L+G)\otimes B)^{T} \\&= (I_{N}\otimes B)\bar{R}^{-1}(I_{N}\otimes B)^{T},
		\end{aligned}
	\end{equation} 
   \begin{equation}
   	\begin{aligned}
   		\bar{D}\Gamma^{-1}\bar{D}^{T}&= ((L+G)\otimes D)((L+G)\otimes I_{m})^{-1}\bar{\Gamma}^{-1}\\&\times((L+G)\otimes I_{m})^{-T}((L+G)\otimes D)^{T} \\&= (I_{N}\otimes D)\bar{R}^{-1}(I_{N}\otimes D)^{T},
   	\end{aligned}
   \end{equation}
		the following equation is obtained 
		\begin{equation}\label{eq:Riccatiweight}
			P\bar{A}+\bar{A}^{T}P +Q -P\hat{B}\bar{R}^{-1}\hat{B}^{T}P +P\hat{D}\bar{\Gamma}^{-1}\hat{D}^{T}P = 0,
		\end{equation}
		where $\hat{B} = I_{N}\otimes B$, $\hat{D} = I_{N} \otimes D$. Since $\bar{A}$, $\bar{R}$, $\bar{\Gamma}$, $\hat{B}$, $\hat{D}$ and $Q$ are all block diagonal matrices, it can be concluded that
		\begin{equation}\label{eq:diagP}
			P = diag\{P_{1}, P_{2},\ldots,P_{N}\}
		\end{equation} is the solution of \eqref{eq:Riccatiweight}, where $P_{i}$, $i =1, 2, \cdots, N$ is the solution of the corresponding local GARE \eqref{eq:localGARE}.
		Then, substituting \eqref{eq:diagP} into the optimal control policy \eqref{eq:u*}, one has 
		\begin{equation}
			u^{*} = -[(L+G)\otimes I_{m1}]^{-1} \bar{R}^{-1}\hat{B}^{T}P\delta.
		\end{equation}
		Considering the diagonal structures of $\bar{R}$, $\hat{B}$ and $P$, the following algebraic equation can be obtained
		\begin{equation}\label{eq:ucoupledoverall}
		\left((L+G)\otimes I_{m} \right)\left[\begin{array}{c}	u_{1}^{*}\\u_{2}^{*}\\ \vdots \\u_{N}^{*}\end{array}\right] =- \left[\begin{array}{c}	K_{1}\delta_{1}\\K_{2}\delta_{2}\\ \vdots\\ K_{N}\delta_{N}\\ \end{array}\right].
	\end{equation}
		where $K_{i}$ is defined as $K_{i} = R_{i}^{-1}B^{T}P_{i}$. Solving $u_{i}^{*}$ from \eqref{eq:ucoupledoverall}, \eqref{eq:ucoupled} can be obtained. 
		
		Similarly, substituting \eqref{eq:diagP} into the worst-case attack \eqref{eq:w*}, one can get the worst-case attack maximizing the global cost function as 
		\begin{equation}
			w^{*} = [(L+G)\otimes I_{m1}]^{-1} \bar{\Gamma}^{-1}\hat{D}^{T}P\delta.
		\end{equation}
	Since $\hat{D}$ is also block-diagonal matrix, one has 
		\begin{equation}\label{eq:wcoupledoverall}
		\left((L+G)\otimes I_{m} \right)\left[\begin{array}{c}	w_{1}^{*}\\w_{2}^{*}\\ \vdots \\w_{N}^{*}\end{array}\right] = \left[\begin{array}{c}	L_{1}\delta_{1}\\L_{2}\delta_{2}\\ \vdots\\ L_{N}\delta_{N}\\ \end{array}\right],
	\end{equation}
   where $L_{i} = 1/\gamma_{i}^{2}D^{T}P_{i}$. Then, solving $w_{i}^{*}$ from \eqref{eq:wcoupledoverall}, one has \eqref{eq:wcoupled}.
	\end{proof}
Based on \eqref{eq:ucoupled}, \eqref{eq:wcoupled}, and \eqref{eq:localGARE}, it is evident that while each agent can compute the min-max strategy gains locally, the resulting strategies remain coupled and cannot be implemented by individual agents directly. As a result, a dynamic average consensus-based algorithm is proposed in the following subsection. This algorithm aims to decouple the min-max strategies presented in \eqref{eq:u*} and \eqref{eq:w*}, thereby enabling their distributed implementation. By leveraging the dynamic average consensus mechanism, the coupling between agents' strategies is effectively addressed, paving the way for individual agents to implement their respective decoupled strategies.
	\subsection{Distributed implementation}
In this subsection, a distributed average consensus-based algorithm (Algorithm 1) is proposed to decouple the coupled min-max strategies. The algorithm utilizes a fast dynamic average consensus-based approach to fuse the local information from each agent and effectively decouple the structures of the optimal control \eqref{eq:ucoupled} and worst-case attack \eqref{eq:wcoupled}.  The proposed algorithm is proven to converge in finite time. Prior to presenting the algorithm, several useful lemmas are introduced.
	\begin{lemma}(see in \cite{robustDAC})\label{lemma:M}
		 For any strongly connected undirected graph of order $n$, we have $M \triangleq\left(I_n-\frac{1}{n} \mathbf{1}_n \mathbf{1}_n^{\top}\right)={L}({L})^{\dagger}=({L})^{\dagger} {L}$, where $(\cdot)^{\dagger}$denotes the generalized inverse.
	\end{lemma}
	\begin{lemma}(see in \cite{1333204})\label{lemma:LtimesX}
		For a connected graph $G$ that is undirected with Laplacian $L$, the following well-known property holds:
		$$
		\min _{\substack{x \neq 0 \\ 1^{\mathrm{T}} \mathrm{x}=0}} \frac{x^T L x}{\|x\|^2}=\lambda_2(L) .
		$$
	\end{lemma}
	\begin{lemma}(see in \cite{graphtheory})\label{lemma:lambda2}
	 For a connected undirect graph $G$ of order $N$, its second Laplacian eigenvalue $\lambda_2(L)$ imposes upper bounds on the diameter and the mean distance of $G$, 
		\begin{equation}
			\operatorname{diam}(G) \geq \frac{4}{N \lambda_2(L)},
		\end{equation}
		where $\operatorname{diam}(\cdot)$ means the diameter operator. The diameter of a graph is the maximum length of the shortest path between any two vertices in the graph. For a connected undirect graph $G$, $\operatorname{diam}(G)<N$. Therefore, one has $ \lambda_{2}(L)\geq \frac{4}{N^{2}}$
	\end{lemma}

	\begin{algorithm}[h!]\label{alg:I}
		\caption{Distributed Decoupled Algorithm} 
		\hspace*{0.02in} {\bf Initialization:} 
		Set $v_{i}(0)$, $w_{i}(0)$, $X_{i}(0)$, $i \in \mathcal{N}$, running time $t_{max}$, 
		and the local gain $\alpha_{i} \geq 1+\eta_{i}\|[K_{i}~L_{i}]\|_{2}\frac{\bar{N}^{4}}{2} $, $\beta_{i} >2\bar{N} \alpha_{i}$\\
		\hspace*{0.02in} 
		{\bf Implement:} 
		\begin{algorithmic}[1]
			\While{$0\leq t \leq t_{max}$} 
			\For{$i = 1 \hspace*{0.05in}  to \hspace*{0.05in}  N$}
			\State {\bf Collect} $L_{i}$, $G_{i}$ and $\delta_{i}(t)$.
			\State {\bf Set} $T_{i} = (L+G)_{i}^{T}(L+G)_{i} \in \mathbb{R}^{N\times N}$.
			\State {\bf Set} $I_{i}(t) = -[(L+G)_{i}^{T}\otimes I_{m_{1}} ]K_{i}\delta_{i}(t) \in \mathbb{R}^{m_{1}N}$.
			\State {\bf Set} $W_{i}(t) = [(L+G)_{i}^{T}\otimes I_{m_{2}} ]L_{i}\delta_{i}(t) \in \mathbb{R}^{m_{2}N}$.
			\State {\bf Set} $\psi_{i}(t) = \left[\begin{array}{c}
				vecs(T_{i})\\I_{i}(t)\\W_{i}(t)
			\end{array}\right]  \in \mathbb{R}^{\bar{m}}$, $\bar{m} = (2m_{1}+2m_{2}+N+1)N/2$.
			\State {\bf Run:}\begin{equation}\label{eq:vup}
				\dot{v}_{i}(t)=-\beta_{i} \operatorname{sgn}\left\{v_{i}(t)-\sum_{j=1}^{N} a_{i j}\left(w_{i}(t)-w_{j}(t)\right)\right\}.
			\end{equation}
			\State {\bf Run:}\begin{equation}\label{eq:wup}
				\dot{w}_{i}(t)=-\alpha_{i} \operatorname{sgn}\left\{\sum_{j=1}^{N} a_{i j}\left(X_{i}(t)-X_{j}(t)\right)\right\}.
			\end{equation}
			\State {\bf Run:}\begin{equation}\label{eq:xup}
				X_{i}(t) = v_{i}(t)+\psi_{i}(t).
			\end{equation}
			\State $$\hat{T}_{i} = ves^{-1} \left([X_{i}(t)]_{1:(N+1)N/2}\right).$$
			\State $$\hat{I}_{i}(t) = [X_{i}(t)]_{1+(N+1)N/2:(2m_{1}+N+1)N/2}.$$
			\State $$\hat{W}_{i}(t) = [X_{i}(t)]_{(2m_{1}+N+1)N/2+1:(2m_{1}+2m_{2}+N+1)N/2}.$$
			\State {\bf Output:}\begin{equation}
				\begin{aligned}
					\hat{u}_{i}(t)& = [(\hat{T}_{i}\otimes I_{m_{1}})^{-1} \hat{I}_{i}(t)]_{(i-1)m_{1}+1:im_{1}},\\
					\hat{w}_{i}(t) &= [(\hat{T}_{i}\otimes I_{m_{2}})^{-1} \hat{w}_{i}(t)]_{(i-1)m_{2}+1:im_{2}}.
				\end{aligned}
			\end{equation}
			\EndFor
			\EndWhile
		\end{algorithmic}
	\end{algorithm}
	
	Algorithm 1 is implemented distributedly, since $(L+G)_{i}$, $K_{i} = R_{i}^{-1}B^{T}P_{i}$, $L_{i} = 1/\gamma_{i}^{2}D^{T}P_{i}$ and $\delta_{i} = \sum_{j \in \mathcal{N}_{i}} a_{i j}\left(x_{i}-x_{j}\right)+g_{i}\left(x_{i}-x_{0}\right)$ are all local information. The outputs $\hat{u}_{i}(t)$ and $\hat{w}_{i}(t)$ are the estimations of $u^{*}_{i}$ in (\ref{eq:ucoupled}) and $w^{*}_{i}$ in \eqref{eq:wcoupled}, respectively. The following theorem will show that the estimate error converges to zero in finite time.

	\begin{theorem}\label{thm:2}
		Given Assumptions 1, 2, 3 and 4, the outputs of Algorithm 1, i.e., $\hat{u}_{i}(t)$ and $\hat{w}_{i}(t)$ converge to $u^{*}_{i}$ in (\ref{eq:ucoupled}) and $w^{*}_{i}$ in \eqref{eq:wcoupled} respectively for all $t\geq t^{*} = \|\tilde{v}(0)\|_{2} +\|X(\|\tilde{v}(0)\|_{2}) - (1_{N}\otimes I_{{\bar{m}}})\bar{\psi}(\|\tilde{v}(0)\|_{2})\|_{2}/\lambda_{2}(L) $, where  $\bar{\psi}(t) = \sum_{i=1}^{N}\psi_{i}(t)/N$. That is to say, for all $t\geq t^{*}$:
		\begin{equation*}
			\hat{u}_{i}(t) = [(\hat{T}_{i}\otimes I_{m_{1}})^{-1} \hat{I}_{i}(t)]_{(i-1)m_{1}+1:im_{1}} = u^{*}_{i}, \quad \forall i \in \mathcal{N}.
		\end{equation*}
		\begin{equation*}
			\hat{w}_{i}(t) = [(\hat{T}_{i}\otimes I_{m_{2}})^{-1} \hat{w}_{i}(t)]_{(i-1)m_{2}+1:im_{2}} = w^{*}_{i}, \quad \forall i \in \mathcal{N}.
		\end{equation*}
	\end{theorem}
	\begin{proof}
		The proof of the theorem is separated in two parts, first we proof that $X_{i}(t)$ converges to $\bar{\psi}(t) = \sum_{i=1}^{N}\psi_{i}(t)/N$ in finite time, then we show that 	$u^{*} = -[(L+G)\otimes I_{m1}]^{-1} \bar{R}^{-1}\hat{B}^{T}P\delta$ and $w^{*} = [(L+G)\otimes I_{m2}]^{-1} \bar{\Gamma}^{-1}\hat{B}^{T}P$ can be express as $(\sum_{i=1}^{N}(T_{i}\otimes I_{m_{1}})^{-1}(\sum_{i=1}^{N}I_{i})$ and $(\sum_{i=1}^{N}(T_{i}\otimes I_{m_{2}})^{-1}(\sum_{i=1}^{N}W_{i})$. Therefore, if $X_{i}\to\bar{\phi}(t) $, $\hat{u}_{i}$ and $\hat{w}_{i}$ converge to $u^{*}_{i}$ and $w^{*}_{i}$, respectively.
		
		{\bf Part 1:} The global form of \eqref{eq:wup}, (\ref{eq:vup}) and (\ref{eq:xup}) can be written as 
		\begin{subequations}\label{eq:overallxv}
			\begin{equation}
				\dot{w}(t) = -\alpha sgn\left\{ (L\otimes I)X(t)\right\},
			\end{equation}
			\begin{equation}
				\dot{v}(t) = -\beta sign\{v(t)-(L\otimes I)w(t)\}
			\end{equation}
			\begin{equation}
				X(t) = v(t) +\psi(t),
			\end{equation}
		\end{subequations}
		where $X(t) = [X^{T}_{1}(t),X^{T}_{2}(t),\ldots,X^{T}_{N}(t)]^{T}$, $v(t) = [v_{1}^{T}(t), v_{2}^{T}(t), \ldots,v_{N}^{T}(t)]^{T}$, $w(t) = [w_{1}^{T}(t), w_{2}^{T}(t), \ldots,w_{N}^{T}(t)]^{T}$, $\psi(t)= [\psi^{T}_{1}(t),\psi^{T}_{2}(t),\ldots,\psi^{T}_{N}(t)]^{T}$, $\alpha =diag\{\alpha_{1}, \alpha_{2}, \ldots, \alpha_{N}\}\otimes I$, $\beta = diag\{\beta_{1},\beta_{2}, \ldots,\beta_{N}\}\otimes I$.
		
		First, it is proved that $\tilde{v} = v(t)-(L\otimes I)w$ converge to zeros in finite time. consider the following Lyapunov function candidate
		\begin{equation}
			V_{1}(\tilde{v}(t)) = \frac{1}{2}\tilde{v}^{T}(t)\tilde{v}(t).
		\end{equation}
		Taking derivative of $V_{1}$, it follows
		\begin{equation}
			\begin{aligned}
				\dot{V}_{1} &= \tilde{v}^{T}\alpha(L\otimes I)sgn\{(L\otimes I)X\} - \|\beta\tilde{v}\|_{1}\\
				&\leq \|\alpha\tilde{v}\|_{1}\|L\|_{1}\|sgn\{(L\otimes I)X\|_{\infty} - \|\beta\tilde{v}\|_{1}.
			\end{aligned}
		\end{equation} 
		Since $\|L\|_{1}<2N-2$ and $\|sgn\{(L\otimes I)X\|_{\infty} \leq 1$, and $\beta_{i}>2\bar{N}\alpha_{i}$, one has
		\begin{equation}
			\dot{V}_{1}\leq -\|\tilde{v}\|_{1}\leq - \|\tilde{v}\|_{2} = -\sqrt{2}\sqrt{V_{1}}.
		\end{equation}
		One can obtianed that 
		\begin{equation}
			\sqrt{V_1(\tilde{v}(t))}\leq \sqrt{V_{1}(\tilde{v}(0))} - \frac{1}{\sqrt{2}}t.
		\end{equation}
		Note that $V_1(\tilde{v}(t))$ is positive-definite, therefore, we have $\tilde{v}(t) = 0$ for all $t\geq \|\tilde{v}(0)\|_{2}$. 
		
		Then, define the tracking error of $X(t)$ as  $\tilde{X}(t) = X(t) - (1_{N}\otimes I)\bar{\psi}(t)$, noting that $\bar{\psi}(t) =(1^{T}_{N}\otimes I)\psi(t)/N $, for all $t>\|\tilde{v}(0)\|_{2}$, one has 
		\begin{equation}\label{eq:errorX}
			\tilde{X}(t) = (L\otimes I)w(t) +(M \otimes I)\psi(t),
		\end{equation}
		where $M =I_{N}-\frac{1}{N}1_{N}1_{N}^{T}$. Take the time derivative of (\ref{eq:errorX}), yields
		\begin{equation}
			\label{eq:derrorX}
			\dot{\tilde{X}}(t) = -\alpha (L\otimes I)sgn\left\{ (L\otimes I)X(t)\right\}+(M\otimes I)\dot{\psi}(t).
		\end{equation}
		Consider the following Lyapunov candidate 
		\begin{equation}
			V(\tilde{X}(t)) = \frac{1}{2}\tilde{X}^{T}(t)\tilde{X}(t).
		\end{equation}
		Taking the time derivative of $	V(\tilde{X}(t))$ and subsitituting (\ref{eq:derrorX}) yields 
		\begin{equation}\begin{aligned}
				\dot{V}(\tilde{X}(t)) = &-\tilde{X}^{T}(t)(L\otimes I)\alpha sgn\left\{ (L\otimes I)X(t)\right\}\\&+\tilde{X}^{T}(t)(M\otimes I)\dot{\psi}(t).
			\end{aligned}
		\end{equation}
		Note that $$(L\otimes I)X(t) = (L\otimes I)\tilde{X}(t),$$ and from Lemma \ref{lemma:M}, $M = L(L)^{\dagger}$, and  $$\tilde{X}^{T}(t)(L\otimes I)\dot{\psi} \leq \|\mathcal{H}(L\otimes I)\tilde{X}(t)\|_{1},$$ where $$\mathcal{H} = 2\bar{N}diag\{\eta_{1}\|[K_{1}~L_{1}]\|_{2},\ldots, \eta_{N}\|[K_{N}~L_{N}]\|_{2}\}.$$ Considering $\|L^{\dagger}\|_{\infty} \leq \frac{N}{\lambda_{2}(L)}$, and $\lambda_{2}(L)\geq \frac{4}{N^{2}}\geq \frac{4}{\bar{N}^{2}}$, one has 
		\begin{equation}
			\dot{V}(\tilde{X}(t))\leq -\|\alpha(L\otimes I)\tilde{X}(t)\|_{1}+ \frac{\bar{N}^{3}}{4}\|\mathcal{H}(L\otimes I)\tilde{X}(t)\|_{1}.
		\end{equation}
		Considering  $\alpha_{i} \geq 1+\eta_{i}\|[K_{i}~L_{i}]\|_{2}\frac{\bar{N}^{4}}{2}$ yields
		\begin{equation}
			\dot{V}(\tilde{X}(t))\leq - \|(L\otimes I_{\bar{m}})\tilde{X}(t)\|_{1}\leq -\|(L\otimes I_{\bar{m}})\tilde{X}(t)\|_{2},
		\end{equation}
		Since $1_{\bar{m}N}^{T}\tilde{X}(t) = 0$, consider Lemma \ref{lemma:LtimesX}, one gets
		\begin{equation}
			\dot{V}(\tilde{X}(t))\leq -\lambda_{2}(L)\sqrt{2}V^{\frac{1}{2}}(\tilde{X}(t)).
		\end{equation}
		From Comparison Lemma, one has for all $t>\|\tilde{v}(0)\|_{2}$, the following holds
		\begin{equation}
			\sqrt{V(\tilde{X}(t))} \leq \sqrt{V(\tilde{X}(\|\tilde{v}(0)\|_{2}))} - \frac{\sqrt{2}}{2}\lambda_{2}(L)(t-\|\tilde{v}(0)\|_{2}).
		\end{equation}
		Therefore, it can be obtained that $\tilde{X}(t) = 0$, for $\forall$ $t\geq t^{*}$, where $t^{*} = \|\tilde{v}(0)\|_{2}+\frac{\tilde{X}(\|\tilde{v}(0)\|_{2})}{\lambda_{2}(L)}$.
		
		{\bf part 2:} Note that $\sum_{i=1}^{N}T_{i} =(L+G)^{T}(L+G)\otimes I_{m} $  , and $\sum_{i=1}^{N}I_{i} = -[(L+G)^{T}\otimes I]K\delta$, $\sum_{i=1}^{N} W_{i} =[(L+G)^{T}\otimes I]L\delta $ since $(L+G)$ is nonsingular, one has 
		\begin{equation}
			u^{*} = (\sum_{i=1}^{N}T_{i})^{-1}(\sum_{i=1}^{N}I_{i}),
		\end{equation}
		\begin{equation}
			u^{*} = (\sum_{i=1}^{N}T_{i})^{-1}(\sum_{i=1}^{N}W_{i}).
		\end{equation}
		Note that $\hat{T}_{i}$, $\hat{I}_{i}$ and $\hat{W}_{i}$ converge to $\frac{1}{N}\sum_{i=1}^{N}T_{i}$, $\frac{1}{N}\sum_{i=1}^{N}I_{i}$ and $\frac{1}{N}\sum_{i=1}^{N}W_{i}$ respectively for all $t>t^{*}$.
		Thus,  $\hat{u}_{i} = u^{*}_{i}$ and $\hat{w}_{i}=w^{*}_{i}$ for all $t>t^{*}$.
	\end{proof}
	\begin{remark}
	Algorithm 1 is proven to enable the distributed computation of the optimal control policy and worst-case attack for the global coalitional min-max problem. While the problem is formulated within a game theory framework, in control applications, our main focus is on the control policy. Hence, the proposed method allows for the distributed determination of each agent's robust control policy, ensuring the satisfaction of the global $L_{2}$ disturbance attenuation condition.
	\end{remark}
	
	\section{Multi-vehicle Cooperation Simulation}
	\subsection{Modeling Formulation}
			\begin{figure}[tpb]
		\centering
		\includegraphics[width=0.6\linewidth]{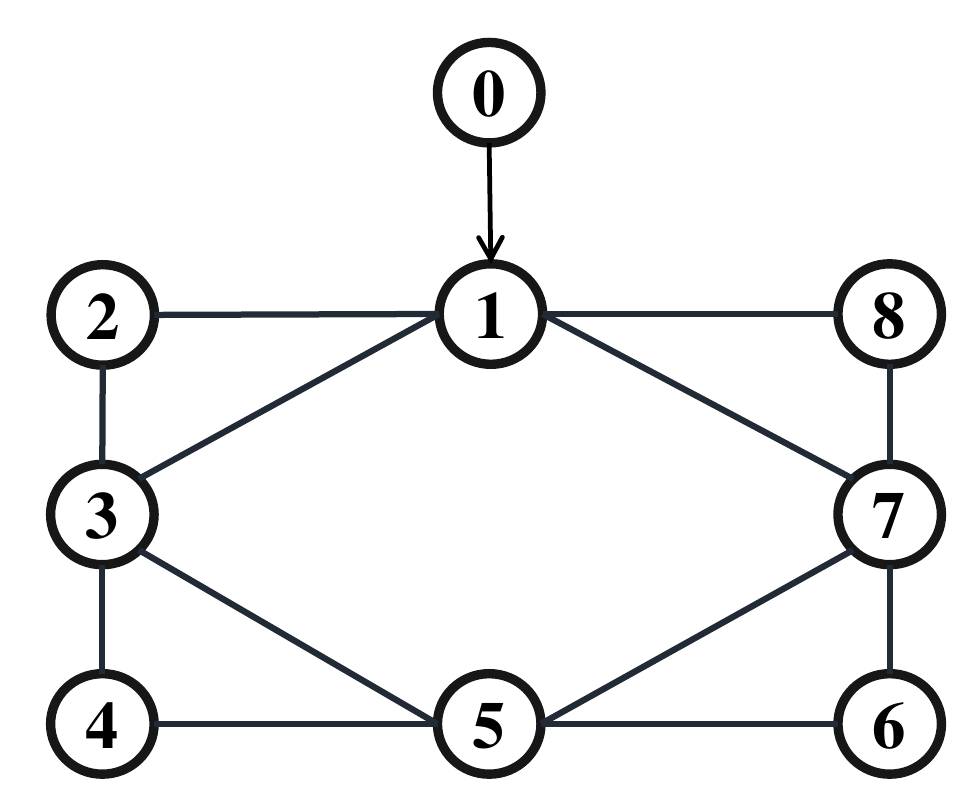}
		\caption{The communication topology.}
		\label{fig:top}
	\end{figure}
Consider the multi-vehicle systems consisting of $N$ wheeled vehicles. The schematic of the $i$th vehicle is shown as Fig \ref{fig:vehicle}. The dynamics of the $i$th vehicle can be described as
	\begin{figure}
		\centering
		\includegraphics[width=0.5\linewidth]{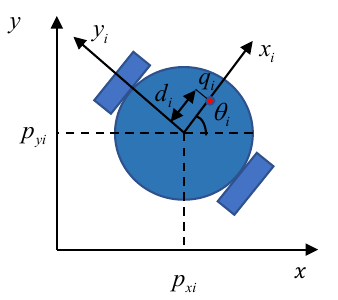}
		\caption{Schematic of the $i$th vehicle.}
		\label{fig:vehicle}
	\end{figure}
	
	\begin{equation}
		\begin{aligned}
			& \dot{p}_{x i}=v_i \cos \theta_i, \dot{p}_{y i}=v_i \sin \theta_i, \dot{\theta}_i=\omega_i, \\
			& \dot{v}_i=F_i / m_i, \dot{\omega}_i=M_i / J_i,
		\end{aligned}
	\end{equation}
where $p_{x i}$, $P_{y,i}$ and $\theta_i$ denote the Cartesian position and orientation of the $i$th vehicle, $v_{i}$ and $\omega_{i}$ are the linear and angular velocities, and $m_{i}$, $J_{i}$ are the mass and moment of inertia of $i$th vehicle. $F_{i}$ and $\tau_{i}$ are the input force and torque of vehicle $i$. Similar to \cite{wangq3}, the dynamics of the $i$th vehicle can be reformulated as follows by feedback linearization at a fixed reference point $q_{i} = [q_{x,i},q_{y,i}]^{T}$ off the center of the vehicle.
	\begin{equation}\label{eq:vehicle_dyn}
		\frac{d}{dt}\left[\begin{array}{c}
			q_{i}\\\dot{q}_{i}
		\end{array}\right] = \left[\begin{array}{cc}
			0_{2} & I_{2} \\0_{2} & 0_{2}
		\end{array}\right]\left[\begin{array}{c}
			q_{i}\\ \dot{q}_{i}
		\end{array}\right]+\left[\begin{array}{c}
			0_{2}\\ I_{2}
		\end{array}\right]u_{i} + \left[\begin{array}{c}
			0_{2}\\ I_{2}
		\end{array}\right]w_{i},
	\end{equation}	
	where $u_{i} = [u_{x,i}~ u_{y,i}]^{T}$ is the control input, $w_{i} = [w_{x,i}~ w_{y,i}]^{T}$ is the energy-bounded attack input. 
	
The most commonly used multi-vehicle cooperation task is formation, as it can be applied to many practical scenarios, such as multi-vehicle handling and transportation. The leader can be considered as virtual; its dynamics are described as
	\begin{equation}\label{eq:leader_virtual}
		\frac{d}{dt}\left[\begin{array}{c}
			q_{0}\\\dot{q}_{0}
		\end{array}\right] = \left[\begin{array}{cc}
			0_{2} & I_{2} \\0_{2} & 0_{2}
		\end{array}\right]\left[\begin{array}{c}
			q_{0}\\ \dot{q}_{0}
		\end{array}\right]
	\end{equation} and its state data is generated by the computer and transmitted to some of the vehicles. Notably, the velocity of the leader $\dot{q}_{0}$ can be set by the computer, allowing the virtual leader to guide the multi-vehicle systems along desired trajectories. Denote $x_{i} = [q^{T}_{i},\dot{q}^{T}_{i}]^{T}$, for formation control, the neighbor error of vehicle $i$ is
	\begin{equation}
		\delta_{i}:=\sum_{j \in \mathcal{N}_{i}} a_{i j}\left(x_{i}-x_{j}\right)+g_{i}\left(x_{i}-x_{0}-x_{c,i}\right),
	\end{equation}
	where $x_{c,i} = [q^{T}_{c,i}~ 0 ~0]^{T}$ is a relative position between the $i$th vehicle and the leader. 
	
	\subsection{Numerical Simulations}
 \begin{figure}[tpb]
	\centering
	\includegraphics[width=0.95\linewidth]{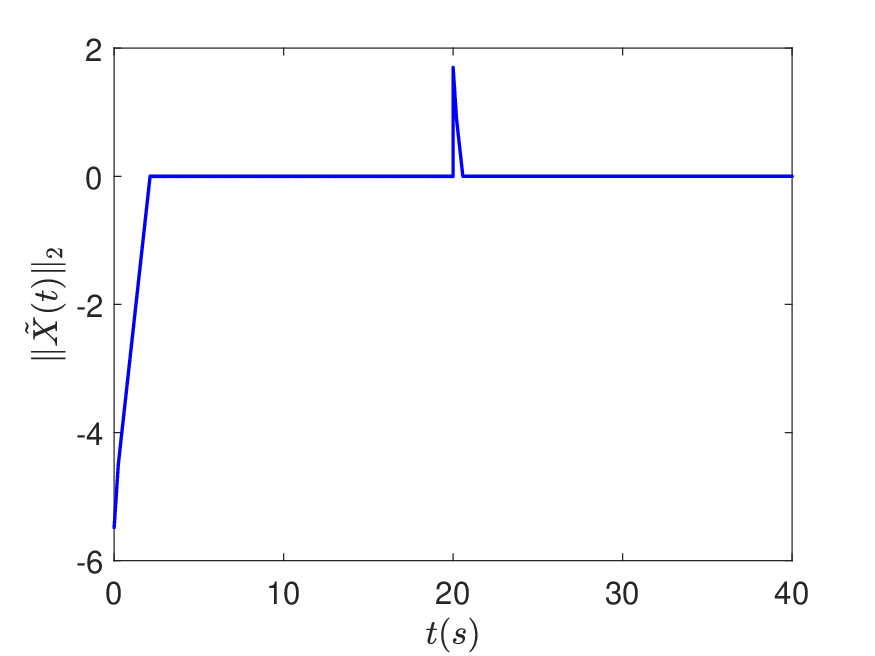}
	\caption{The convergence of Algorithm 1.}
	\label{fig:tildex}
\end{figure}
	Consider 8 vehicles with dynamics described as \eqref{eq:vehicle_dyn} and one leader described as \eqref{eq:leader_virtual}. The $i$th vehicle's position is initiated as $[5+3*cos(\frac{(i-1)\pi}{4})~ 5+3*sin(\frac{(i-1)\pi}{4})]^{T}$. The velocity of each vehicle is $[0~ 0 ]^{T}$. The leader's initial position and velocity are $[0~0~-0.5~-0.1]^{T}$. The desired formation pattern is defined as
	\begin{equation}
		q_{c,i} =\left\{\begin{array}{c} \left[0.9*cos(\frac{(i-1)\pi }{4})~ 1.2*sin(\frac{(i-1)\pi }{4})\right]^{T},~ t<20,  \\ \left[2*cos(\frac{(i-1)\pi }{4})~ 3*sin(\frac{(i-1)\pi }{4})\right]^{T},~ t\geq 20. 
		\end{array}	\right.
	\end{equation}
	The communication topology of vehicles is illustrated as \figurename \ref{fig:top}. The local consensus gain in Algorithm 1 is $\alpha_{i} = 2100$, $\beta_{i} = 34000$. The local weighting matrices are chosen as $Q_{i} = 10I_{4}$, $R_{i} = I_{2}$ and the local $L_{2}$ gain is $\gamma_{i} = 2$.

 \figurename \ref{fig:tildex} shows the consensys error of Algorithm 1, the result demonstrates that Algorithm can converge very fast meet the decoupling requirements.

	\figurename \ref{fig:formation} illustrates the trajectories of 8 vehicles, it can be obatained that the 8 vehicles can effectively form a formatiion and transform the formation by reset their relative positions with the leader.
	\begin{figure}[tpb]
		\centering
		\includegraphics[width=0.95\linewidth]{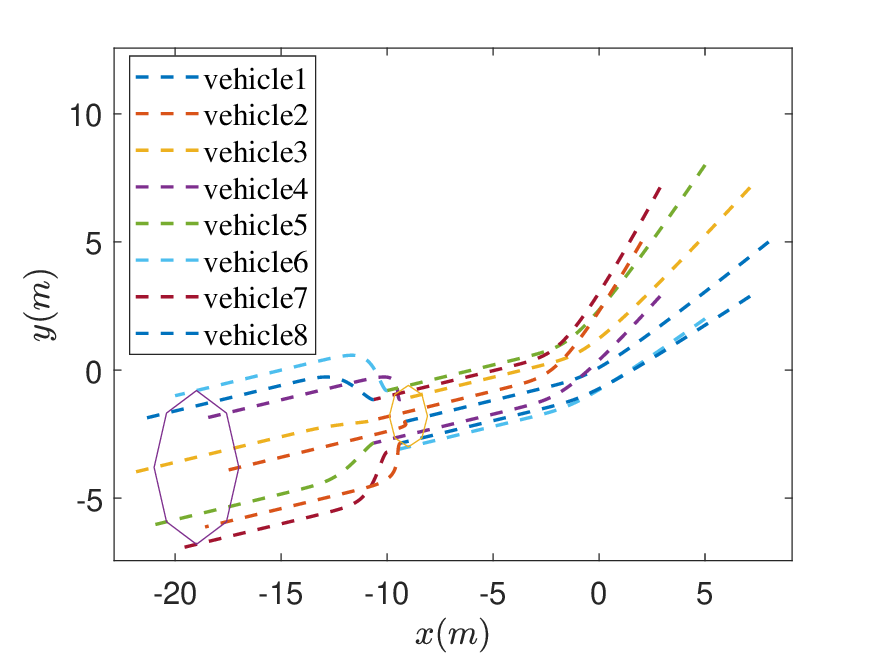}
		\caption{The trajectories of 8 vehicles.}
		\label{fig:formation}
	\end{figure}

	\begin{figure}[thpb]
		\centering
		\includegraphics[width=0.95\linewidth]{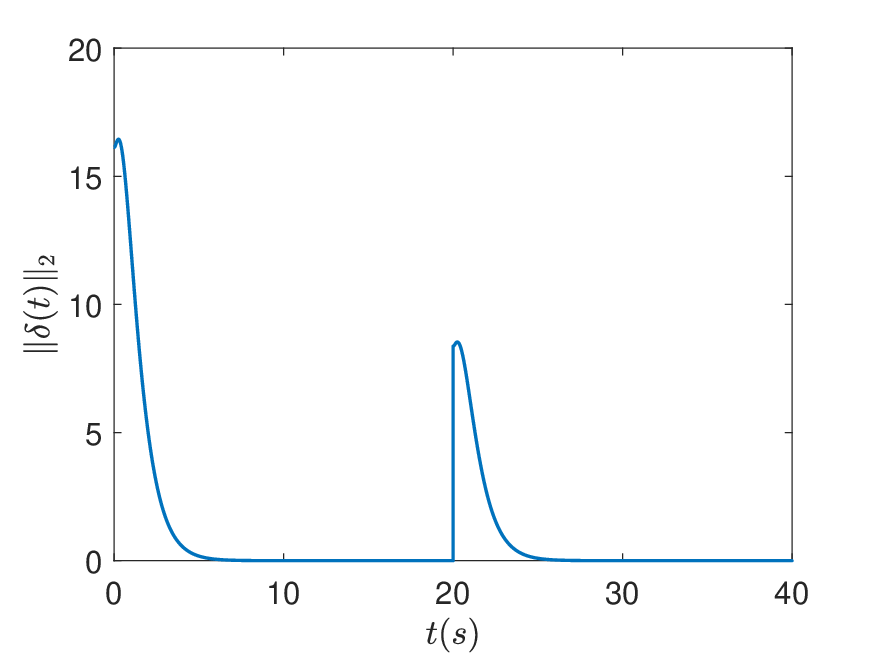}
		\caption{The evelotion of the norm of the global tracking error.}
		\label{fig:errorconverge}
	\end{figure}
\begin{figure}[htpb]
	\centering
	\includegraphics[width=0.9\linewidth]{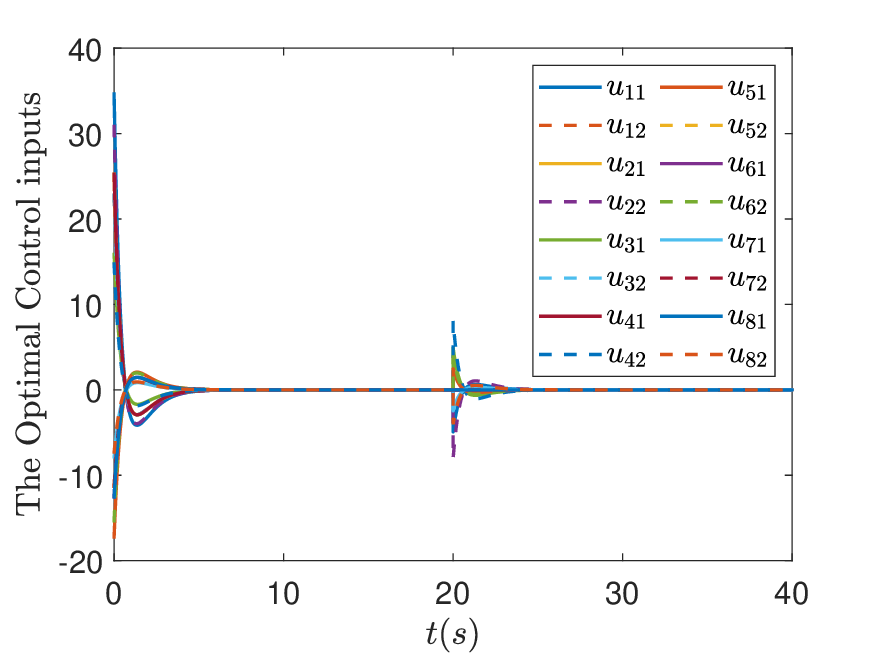}
	\caption{The min-max control inputs of 8 vehicles.}
	\label{fig:u}
\end{figure}
\begin{figure}
	\centering
	\includegraphics[width=0.9\linewidth]{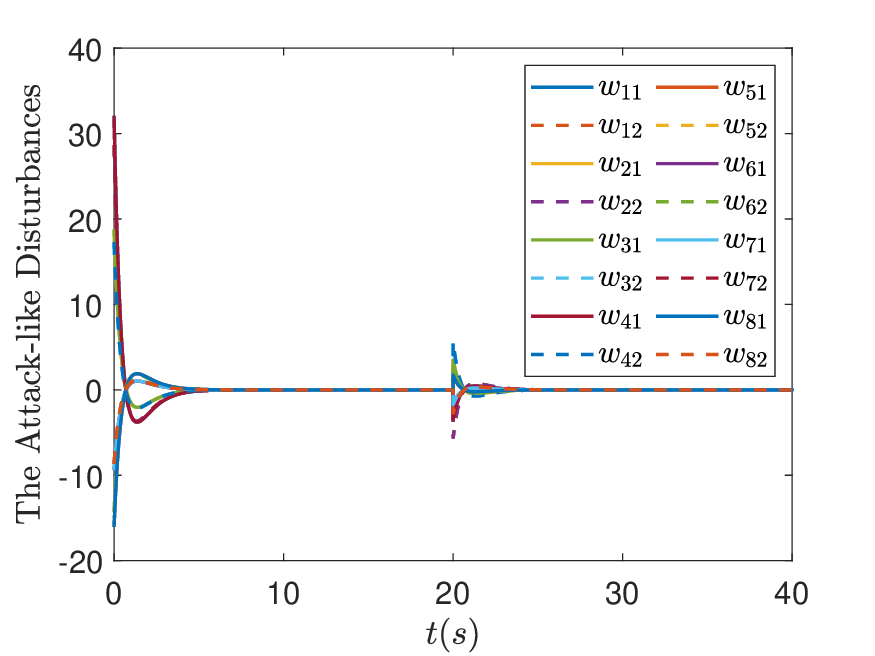}
	\caption{The worst-cast energy-bounded attack inputs of 8 vehicles.}
	\label{fig:w}
\end{figure}

The convergence result of the global tracking error is shown in \figurename \ref{fig:errorconverge}. The error norm undergoes a mutation because the formation changes at 20 seconds. As shown in \figurename \ref{fig:formation} and \ref{fig:errorconverge}, the vehicles can achieve the desired formation in a short time. When the formation pattern changes, the vehicles can quickly adapt to the change and form the new formation.

The control inputs of 8 vehicles, obtained by solving the min-max problem \eqref{eq:minmax-prob} are illustrated in \figurename \ref{fig:u}, and the worst-case attacks of the 8 vehicles, solved by the min-max problem \eqref{eq:minmax-prob} are illustrated in \figurename \ref{fig:w}.

\section{conclusion}
This paper focuses on the robust leader-following consensus problem in multi-agent systems with attacks and its application to multi-vehicle cooperative formation. The main contribution of the paper lies in formulating the robust leader-following consensus control problem as a global coalitional min-max game problem and proving the effectiveness of the obtained controller in solving the consensus problem. In particular, the paper introduces a decentralized computation strategy and implements a distributed decoupling algorithm to grapple with the intricate coupling challenges encountered during controller computation and implementation. This approach boasts two pivotal advantages. Firstly, it presents user-friendly global cost weights, empowering users to adjust these weights based on their specific performance criteria. Secondly, it ensures that each agent is only required to solve a local, low-dimensional GARE to derive the min-max strategy. This efficient approach effectively avoids the curse of dimensionality, which often affects systems with a growing number of agents. The validity and efficacy of the proposed approach are substantiated through a series of simulation examples that illustrate the correctness and practical utility of the introduced methodology.

In terms of future research directions, extending the method to nonlinear systems holds promise. This would enable the application of the proposed approach to a wider range of real-world systems. Furthermore, integrating data-driven techniques with the proposed method could enhance its performance and applicability in scenarios with limited prior knowledge or uncertain system dynamics.

\end{document}